\documentclass{article}
\usepackage{fullpage,amsmath,amsthm,amssymb}

\newcommand{\Alg}{\mathcal{A}}

\newcommand{\bG}{\mathbf{G}}

\newcommand{\cT}{\mathcal{T}}

\newcommand{\bX}{\mathbf{X}}
\newcommand{\cM}{\mathcal{M}}

\newcommand{\bp}{\mathbf{p}}

\newcommand{\bv}{\mathbf{v}}

\newcommand{\allOne}{\mathbf{1}}
\DeclareMathOperator{\dist}{dist}

\newcommand{\cN}{\mathcal{N}}

\newtheorem{theorem}{Theorem}
\newtheorem{lemma}{Lemma}
\newtheorem{corollary}{Corollary}

\title{Sublinear Time Shortest Path in Expander Graphs}

\author{Noga Alon\thanks{\texttt{nalon@math.princeton.edu}. Supported in part by NSF grant DMS-2154082
and BSF grant 2018267.}\\Princeton University \and Allan Gr\o nlund\thanks{\texttt{ag@kvantify.dk}.}\\Kvantify \and S\o ren Fuglede J\o rgensen\thanks{\texttt{sfj@kvantify.dk}.}\\Kvantify \and Kasper Green Larsen\thanks{\texttt{larsen@cs.au.dk}. Supported by a DFF Sapere Aude Research Leader Grant No. 9064-00068B.}\\Aarhus University \& Kvantify}

\date{}

\begin{document}

\maketitle

\begin{abstract}
Computing a shortest path between two nodes in an undirected unweighted graph is among the most basic algorithmic tasks. Breadth first search solves this problem in linear time, which is clearly also a lower bound in the worst case. However, several works have shown how to solve this problem in sublinear time in expectation when the input graph is drawn from one of several classes of random graphs. In this work, we extend these results by giving sublinear time shortest path (and short path) algorithms for expander graphs. We thus identify a natural deterministic property of a graph (that is satisfied by typical random regular graphs) which suffices for sublinear time shortest paths. The algorithms are very simple, involving only bidirectional breadth first search and short random walks. We also complement our new algorithms by near-matching lower bounds.
\end{abstract}

\section{Introduction}
Computing shortest paths in an undirected unweighted graph is among the most fundamental tasks in graph algorithms. In the single source case, the textbook breadth first search (BFS) algorithm computes such shortest paths in $O(m+n)$ time in a graph with $n$ nodes and $m$ edges. Linear time is clearly also a lower bound on the running time of any algorithm that is correct on all input graphs, even if we only consider computing a shortest $s$-$t$ path for a pair of nodes $s,t$, and not the shortest path from $s$ to all other nodes. Initial intuition might also suggest that linear time is necessary for computing the shortest path between two nodes $s,t$ in a random graph drawn from any reasonable distribution, such as an Erd\H{o}s-R\'{e}nyi random graph or a random $d$-regular graph. However, this intuition is incorrect and there exists an algorithm with a sublinear expected running time for many classes of random graphs~\cite{hyperbolic, kadabra, oldbibfs}. Moreover, the algorithm is strikingly simple! It is merely the popular practical heuristic of bidirectional BFS~\cite{biBFS}. In bidirectional BFS, one simultaneously runs BFS from the source $s$ and destination $t$, expanding the two BFS trees by one layer at a time. If the input graph is e.g. an Erd\H{o}s-R\'{e}nyi random graph, then it can be shown that the two BFS trees have a node in common after exploring only $O(\sqrt{n})$ nodes in expectation. If the node $v$ is first to be explored in both trees, then the path from $s \to v \to t$ in the two BFS trees form a shortest path between $s$ and $t$. The fact that only $O(\sqrt{n})$ nodes need to be explored intuitively follows from the birthday paradox and the fact that the nodes nearest to $s$ and $t$ are uniform random in an Erd\H{o}s-R\'{e}nyi random graph (although not completely independent). 
Note that for sublinear time graph algorithms to be meaningful, we assume that we have random access to the nodes and their neighbors. More concretely, we assume the nodes are indexed by integers $[n] = \{1,\dots,n\}$ and that we can query for the number of nodes adjacent to a node $v$, as well as query for the $j$'th neighbor of a node $v$. We remark that several works have also extended the bidirectional BFS heuristic to weighted input graphs and/or setups where heuristic estimates of distances between nodes and the source or destination are known~\cite{biBFS, bisearch, bisearchagain}. There are also works giving sublinear time algorithms for other natural graph problems under the assumption of a random input graph~\cite{hochbaum}.

A caveat of the previous works that give provable sublinear time shortest path algorithms, is that they assume a random input graph. In this work, we identify "deterministic" properties of graphs that may be exploited to obtain sublinear time $s$-$t$ shortest path algorithms. Concretely, we study shortest paths in expander graphs. An $n$-node $d$-regular (all nodes have degree $d$) graph $G$, is an $(n,d,\lambda)$-graph if the eigenvalues $\lambda_1 \geq \cdots \geq \lambda_n$ of the corresponding adjacency matrix $A$ satisfies $\max_{i \neq 1}|\lambda_i| \leq \lambda$. Note that the eigenvalues are real since $A$ is symmetric and real. We start by presenting a number of algorithmic results when the input graph is an expander.

\paragraph{Shortest $s$-$t$ Path.}
Our first contribution demonstrates that the simple bidirectional BFS algorithm efficiently computes the shortest path between most pairs of nodes $s, t$ in an expander:
\begin{theorem}
\label{thm:bibfsExp}
    If $G$ is an $(n,d,\lambda)$-graph, then for every node $s \in G$, every $0 < \delta < 1$, it holds for at least $(1-\delta)n$ nodes $t$, that bidirectional BFS between $s$ and $t$, finds a shortest $s$-$t$ path after visiting $O((d-1)^{\lceil (1/4) \lg_{d/\lambda}(n/\delta) \rceil})$ nodes.
\end{theorem}
While the bound in Theorem~\ref{thm:bibfsExp} on the number of nodes visited may appear unwieldy at first, we note that it simplifies significantly for natural values of $d$ and $\lambda$. For instance, an $(n,d,\lambda)$-graph is Ramanujan if $\lambda \leq 2\sqrt{d-1}$. For Ramanujan graphs, and more generally for graphs with $\lambda = O(\sqrt{d})$, the bound in Theorem~\ref{thm:bibfsExp} simplifies to near-$\sqrt{n}$:
\begin{corollary}
\label{cor:biBFS}
    If $G$ is an $(n,d,O(\sqrt{d}))$-graph, then for every node $s \in G$, every $0 < \delta < 1$, it holds for at least $(1-\delta)n$ nodes $t$, that bidirectional BFS between $s$ and $t$, finds a shortest $s$-$t$ path after visiting $O((n/\delta)^{1/2 + O(1/\ln d)})$ nodes.
\end{corollary}
We also demonstrate that the bound can be tightened even further for Ramanujan graphs:
\begin{theorem}
\label{thm:biRamanujan}
    If $G$ is a $d$-regular Ramanujan graph where $d \geq 3$, then for every node $s \in G$, it holds for at least $(1-o(1))n$ nodes $t$, that bidirectional BFS between $s$ and $t$, finds a shortest $s$-$t$ path after visiting $O( \sqrt{n}\cdot  \ln^{3/2}(n))$ nodes.
\end{theorem}

\paragraph{Short $s$-$t$ Path.}
One drawback of bidirectional BFS in expanders, is that it is only guaranteed to find a shortest path efficiently for \emph{most} pairs of nodes $s,t$. Motivated by this shortcoming, we also present a simple randomized algorithm for finding a short, but not necessarily shortest, $s$-$t$ path. For any parameter $0 < \delta < 1$, the algorithm starts by growing a BFS tree from $s$ until $\Theta(\sqrt{n \ln(1/\delta)})$ nodes have been explored. It then performs $O(\sqrt{n \ln(1/\delta)}/\lg_{d/\lambda}(n))$ random walks starting at $t$. Each of these random walks run for $O(\lg_{d/\lambda}(n))$ steps. If any of these walks discover a node in the BFS tree, it has found an $s$-$t$ path of length $O(\lg_{d/\lambda}(n))$.

We show that this BFS + Random Walks algorithm has a high probability of finding an $s$-$t$ path:
\begin{theorem}
\label{thm:path}
    If $G$ is an $(n,d,\lambda)$-graph with $\lambda \leq d/2$, then for every pair of nodes $s,t$, every $0 < \delta < 1$, it holds with probability at least $1-\delta$, that BFS + Random Walks between $s$ and $t$, finds an $s$-$t$ path of length $O(\lg_{d/\lambda}(n))$ while visiting $O(\sqrt{n \ln(1/\delta)})$ nodes.
\end{theorem}

\paragraph{Lower Bounds.}
While bidirectional BFS, or BFS + Random Walks, are natural algorithms for finding $s$-$t$ paths efficiently, it is not a priori clear that better strategies do not exist. One could e.g. imagine sampling multiple nodes in an input graph, growing multiple small BFS trees from the sampled nodes and somehow use this to speed up the discovery of an $s$-$t$ path. To rule this approach out, we complement the algorithms presented above with lower bounds. For proving lower bounds, we consider distributions over input graphs and show that any algorithm that explores few nodes fails to find an $s$-$t$ path with high probability in such a random input graph. As Erd\H{o}s-R\'{e}nyi random graphs (with large enough edge probability) and random $d$-regular graphs are both expanders with good probability, we prove lower bounds for both these random graph models. The distribution of an Erd\H{o}s-R\'{e}nyi random graph on $n$ nodes is defined from a parameter $0 < p < 1$. In such a random graph, each edge is present independently with probability $p$. A random $d$-regular graph on the other hand, is uniform random among all $n$-node graphs where every node has degree $d$. 

Our lower bounds hold even for the problem of reporting an arbitrary path connecting a pair of nodes $s, t$, not just for reporting a short/shortest path. Furthermore, our lower bounds are proved in a model where we allow node-incidence queries. A node-incidence query is specified by a node index $v$ and is returned the set of all edges incident to $v$. Our first lower bound holds for Erd\H{o}s-R\'{e}nyi random graphs:
\begin{theorem}
\label{thm:lbErdos}
    Any (possibly randomized) algorithm for reporting an $s$-$t$ path in an Erd\H{o}s-R\'{e}nyi random graph, where edges are present with probability $p \geq 1.5 \ln(n)/n$, either makes $\Omega(1/(p\sqrt{n}))$ node-incidence queries or outputs a valid path with probability at most $o(1)+p$.
\end{theorem}
Note that the lower bound assumes $p \geq 1.5 \ln(n)/n$. This is a quite natural assumption since for $p \ll \ln(n)/n$, the input graph is disconnected with good probability. The concrete constant $1.5$ is mostly for simplicity of the proof. We remark that the additive $p$ in the success probability is tight as an algorithm always reporting the direct path consisting of the single edge $(s,t)$ is correct with probability $p$. Also observe that the number of edges discovered after $O(1/(p\sqrt{n}))$ node-incidence queries is about $O(pn/(p\sqrt{n})) = O(\sqrt{n})$ since each node has $p(n-1)$ incident edges in expectation. 

For the case of random $d$-regular graphs, we show the following lower bound for constant degree $d$:
\begin{theorem}
\label{thm:lbRegular}
    Any (possibly randomized) algorithm for reporting an $s$-$t$ path in a random $d$-regular graph with $d=O(1)$, either makes $\Omega(\sqrt{n})$ node-incidence queries or outputs a valid path with probability at most $o(1)$.
\end{theorem}
We remark that a random $d$-regular graph is near-Ramanujan with probability $1-o(1)$~as proved in \cite{friedman08}, confirming a conjecture raised in \cite{alon86}. A near-Ramanujan graph is an $(n,d,\lambda)$-expander with $\lambda \leq 2\sqrt{d-1}+o(1)$. Thus our upper bounds in Theorem~\ref{thm:bibfsExp} and Theorem~\ref{thm:path} nearly match this lower bound.

\paragraph{Overview.}
In Section~\ref{sec:upper}, we present our upper bound results and prove the claims in Theorem~\ref{thm:bibfsExp} and Theorem~\ref{thm:path}. The upper bounds are all simple algorithms and also have simple proofs using well-known facts about expanders.

In Section~\ref{sec:lower}, we prove our lower bounds. These proofs are more involved and constitute the main technical contributions of this work.

\section{Upper Bounds}
\label{sec:upper}
In the following, we present and analyse simple algorithms for various $s$-$t$ reachability problems in expander graphs.

\subsection{Shortest Path}
Let $G$ be an $(n,d,\lambda)$-graph and consider the following bidirectional BFS algorithm for finding a shortest path between a pair of nodes $s,t$: grow a BFS tree $\cT_s$ from $s$ and a BFS tree $\cT_t$ from $t$ simultaneously. In each iteration, the next layer of $\cT_s$ and $\cT_t$ is computed and as soon as a node $v$ appears in both trees, we have found a shortest path from $s$ to $t$, namely the path $s \to v \to t$ in the two BFS trees.

We show that this algorithm is efficient for most pairs of nodes $s, t$ as claimed in Theorem~\ref{thm:bibfsExp}.

To prove Theorem~\ref{thm:bibfsExp}, we show that in any $(n,d,\lambda)$-graph $G$, it holds for every node $s \in G$ that most other nodes have a small distance to $s$. Concretely, we show the following
\begin{lemma}
\label{lem:mostclose}
    If $G$ is an $(n,d,\lambda)$-graph, then for every node $s \in G$, it holds for every $0 < \delta < 1$ that there are no more than $\delta n$ nodes with distance more than $(1/2)\lg_{d/\lambda}(n/\delta)$ from $s$.
\end{lemma}
Theorem~\ref{thm:bibfsExp} now follows from Lemma~\ref{lem:mostclose} by observing that for a pair of nodes $s,t$ of distance $k$ in an $(n,d,\lambda)$-graph, the bidirectional searches will meet after expanding for $\lceil k/2 \rceil$ steps from $s$ and $t$. Since each node explored during breadth first search has at most $d-1$ neighbors outside the previously explored tree, it follows that the total number of nodes visited is $O((d-1)^{\lceil k/2 \rceil})$. Since it holds for every $s \in G$ that $\dist(s,t) \leq (1/2)\lg_{d/\lambda}(n/\delta)$ for a $1-\delta$ fraction of all other nodes $t$, the conclusion follows.

Corollary~\ref{cor:biBFS} follows from Theorem~\ref{thm:bibfsExp} by observing that for $\lambda = O(\sqrt{d})$, we have $(1/4) \lg_{d/\lambda}(n/\delta) =  (1/2)\lg_{\Omega(d)}(n/\delta)$. Noting that $\lg_{\Omega(d)}(n/\delta) = \ln(n/\delta)/(\ln(d) - O(1)) = (1+O(1/\ln d)) \lg_{d-1}(n/\delta)$, the conclusion follows.

What remains is to prove Lemma~\ref{lem:mostclose}. While the contents of the lemma is implicit in previous works, we have not been able to find a reference explicitly stating this fact. We thus provide a simple self-contained proof building on Chung's~\cite{Chung1989DiametersAE} proof that the diameter of an $(n,d,\lambda)$-graph is bounded by $\lceil \lg_{d/\lambda} n \rceil$. 

\begin{proof}[Proof of Lemma~\ref{lem:mostclose}]
Let $A$ be the adjacency matrix of an $(n,d,\lambda)$-graph $G$. Letting $d=\lambda_1 \geq \lambda_2 \geq \cdots \geq \lambda_n$ denote the (real-valued) eigenvalues of the real symmetric matrix $A$, we may write $A$ in its spectral decomposition $A=U \Sigma U^T$ with $\lambda_1,\dots,\lambda_n$ being the diagonal entries of the diagonal matrix $\Sigma$. By definition, we have $\max\{\lambda_2, |\lambda_n|\} = \lambda$. 

Notice that $(A^k)_{s,t}$ gives the number of length-$k$ paths from node $s$ to node $t$ in $G$. Furthermore, we have $A^k = U \Sigma^k U^T$. Now let $s$ be an arbitrary node of $G$ and let $Z \subseteq [n]$ denote the subset of columns $t$ such that $(A^k)_{s,t} = 0$. The eigenvalues of $A^k$ are $\lambda_1^k,\dots,\lambda_n^k$ and the all-1's vector $\allOne$ is an eigenvector corresponding to $\lambda_1$. Let $\allOne_Z$ denote the indicator for the set $Z$, i.e. the coordinates of $\allOne_Z$ corresponding to $t \in Z$ are $1$ and the remaining coordinates are $0$. By definition of $Z$, we have that $e_s^T A^k \allOne_Z = 0$. At the same time, we may write $\allOne_Z = (|Z|/n)\allOne + \beta u$ where $u$ is a unit length vector orthogonal to $\allOne$ and $\beta = \sqrt{|Z|-|Z|^2/n}$. Hence
\begin{eqnarray*}
    0 &=& e^T_s A^k \allOne_Z \\
    &=& e^T_s A^k ((|Z|/n)\allOne + \beta u) \\
    &=& e^T_s \lambda_1^k (|Z|/n) \allOne + \beta e^T_s A^k u \\
    &\geq& d^k|Z|/n - \beta \cdot \|e_s\| \cdot \|A^k u\| \\
    &\geq& d^k |Z|/n - \beta \lambda^k.
\end{eqnarray*}
From this we conclude $|Z| \leq (\lambda/d)^k n \beta \leq (\lambda/d)^k n \sqrt{|Z|}$, implying $|Z| \leq (\lambda/d)^{2k} n^2$. For $k = (1/2)\lg_{d/\lambda}(n/\delta)$, this is $|Z| \leq \delta n$.
\end{proof}

For the special case of Ramanujan graphs, Theorem~\ref{thm:biRamanujan} claims an even stronger result than Theorem~\ref{thm:bibfsExp}. Recall that an $(n,d,\lambda)$-graph is Ramanujan if it satisfies that $\lambda \leq 2\sqrt{d-1}$. To prove Theorem~\ref{thm:biRamanujan} we make use of the following concentration result on distances in Ramanujan graphs:
\begin{theorem}[\cite{Lubetzky2015CutoffOA}]
\label{thm:ramanujan}
Let $G$ be a $d$-regular Ramanujan graph on $n$ nodes, where $d \geq 3$. Then for every node $s \in G$ it holds that
\[
\lvert\{ t \in G : \lvert\dist(s,t) - \lg_{d-1}n\rvert > 3 \lg_{d-1} \lg n \}\rvert = o(n).
\]
\end{theorem}
Using Theorem~\ref{thm:ramanujan}, we conclude that for every node $s \in G$, it holds for $(1-o(1))n$ choices of $t$ that $\dist(s,t) \leq \lg_{d-1} n + 3 \lg_{d-1} \lg n$. The middle node $v$ on a shortest path from $s$ to $t$ thus has distance at most $k = \lceil (\lg_{d-1} n + 3 \lg_{d-1} \lg n)/2 \rceil \leq (1/2)\lg_{d-1} n + (3/2) \lg_{d-1} \lg n + 1$ from $s$ and $t$. Since the nodes in a layer $\ell$ of a BFS tree in a $d$-regular graph $G$ has at most $d-1$ neighbors in layer $\ell+1$, we conclude that the two BFS trees $\cT_s$ and $\cT_t$ contain at most $O((d-1)^k) \leq O(\sqrt{n} \cdot \ln^{3/2}(n))$ nodes each upon termination. Note that the same 
proof shows how to find a shortest path in time $n^{1/2+o(1)}$ between most pairs of vertices $s$ and $t$ in near Ramanujan graphs, as it is also proved in \cite{Lubetzky2015CutoffOA} that in such graphs, for every node $s$ there are
only $o(n)$ nodes $t$ of distance exceeding $(1+o(1)) \lg_{d-1} n$ from $s$.

\subsection{Connecting Path}
In the following, we analyse our algorithm, BFS + Random Walks, for finding a short $s$-$t$ path in an $(n,d,\lambda)$-graph. The algorithm is parameterised by an integer $k \geq \sqrt{n}$ and is as follows: First, run BFS from $s$ until $k$ nodes have been discovered. Call the set of discovered nodes $V_s$. Next, run $\tau = k/(3 \lg_{d/\lambda}(n))$ random walks $\bp_1,\dots,\bp_\tau$ from $t$, with each random walk having a length of $3 \lg_{d/\lambda}(n)$. If any of the random walks intersects $V_s$, we have found an $s$-$t$ path of length $O(\lg_{d/\lambda}(n))$ as the paths $\bp_i$ have length $O(\lg_{d/\lambda}(n))$ and the diameter, and hence the depth of the BFS tree, in an $(n,d,\lambda)$-graph is at most $\lceil \lg_{d /\lambda}(n) \rceil$~\cite{Chung1989DiametersAE}.

To analyse the success probability of the algorithm, we bound the probability that all paths $\bp_i$ avoid $V_s$. For this, we use the following two results

\begin{theorem}[\cite{hoory06}]
\label{thm:fastmix}
Let $G$ be an $(n,d,\lambda)$-graph. For any two nodes $s,t$ in $G$, the probability $p^k_{s,t}$ that a random walk starting in $s$ and of length $k$ ends in the node $t$, satisfies $|1/n - p^k_{s,t}| \leq (\lambda/d)^k$.
\end{theorem}

\begin{theorem}[\cite{alonDerand}]
\label{thm:hardavoid}
Let $G$ be an $(n,d,\lambda)$-graph and let $W$ be a set of $w$ vertices in $G$ and set $\mu = w/n$. Let $P(W,k)$ be the total number of length $k$ paths ($k+1$ nodes) that stay in $W$. Then
\[
P(W,k) \leq wd^{k}(\mu + (\lambda/d) (1-\mu))^{k}.
\]
\end{theorem}

Now consider one of the length $3 \lg_{d/\lambda}(n)$ random walks $\bp = \bp_i$ starting in $t$. To show that it is likely that the path intersects $V_s$, we split the random walk $\bp = (t, \bv_1, \dots,\bv_{3\lg_{d/\lambda}(n) + 1})$ into two parts, namely the first $2\lg_{d/\lambda}(n)$ steps $\bp^{(1)} = (t,\bv_1,\dots,\bv_{2\lg_{d/\lambda}(n)+1})$ and the remaining $\lg_{d/\lambda}(n)$ steps $\bp^{(2)} = (\bv_{2\lg_{d/\lambda}(n)+1},\dots,\bv_{3\lg_{d/\lambda}(n)+1})$. Note that we let the last node $e(\bp^{(1)}) = \bv_{2\lg_{d/\lambda}(n)+1}$ in $\bp^{(1)}$ equal the first node $s(\bp^{(2)}) = \bv_{2\lg_{d/\lambda}(n)+1}$ in $\bp^{(2)}$. We use $\bp^{(1)}$ to argue that $\bp^{(2)}$ has a near-uniform random starting node. We then argue that $\bp^{(2)}$ intersects $V_s$ with good probability.

By Theorem~\ref{thm:fastmix}, it holds for any node $r \in G$ that $\Pr[e(\bp^{(1)}) = r] \leq 1/n + 1/n^2$. Next, conditioned on $e(\bp^{(1)})=r$, the path $\bp^{(2)}$ is uniform random among the $d^{\lg_{d/\lambda}(n)}$ length $\lg_{d/\lambda}(n)$ paths starting in $r$. It follows that for any fixed path $p$ of length $\lg_{d/\lambda}(n)$ in $G$, we have $\Pr[\bp^{(2)} = p] \leq \Pr[e(\bp^{(1)}) = s(p)] d^{-\lg_{d/\lambda}(n)} \leq (1/n + 1/n^2)d^{-\lg_{d/\lambda}(n)}$. Now by Theorem~\ref{thm:hardavoid} with $W=V(G) \setminus V_s$ and assuming $\lambda \leq d/2$, there are at most $n d^{\lg_{d/\lambda}(n)}((1-k/n) + (\lambda/d)(k/n))^{\lg_{d/\lambda}(n)} \leq n d^{\lg_{d/\lambda}(n)} (1-k/(2n))^{\lg_{d/\lambda}(n)} \leq n d^{\lg_{d/\lambda}(n)} \exp(-\lg_{d/\lambda}(n)k/(2n))$ paths in $G$ that stay within $V(G) \setminus V_s$. A union bound over all of them implies that the probability that $\bp^{(2)}$ avoids $V_s$ is at most
\[
(1/n+1/n^2)d^{-\lg_{d/\lambda}(n)} n d^{\lg_{d/\lambda}(n)} \exp(-\lg_{d/\lambda}(n)k/(2 n)) \leq \exp(-\lg_{d/\lambda}(n)k/(2n) + 1/n).
\]
Since the $\tau = k/(3 \lg_{d/\lambda}(n))$ random walks $\bp_1,\dots,\bp_\tau$ are independent, we conclude that the probability they all avoid $V_s$ is no more than
\[
\exp(-k^2/(6n) + k/(3 \lg_{d/\lambda}(n)n)).
\]
Letting $k = \sqrt{7 n \ln(1/\delta)}$ and assuming $n$ is at least some sufficiently large constant, we have that at least one path $\bp_i$ intersects $V_s$ with probability at least $1-\delta$. This completes the proof of Theorem~\ref{thm:path}.

\section{Lower Bounds}
\label{sec:lower}
In this section, we prove lower bounds on the  number of queries made by any algorithm for computing an $s$-$t$ path in a random graph. Our query model allows \emph{node-incidence queries}. Here the $n$ nodes of a graph $G$ are assumed to be labeled by the integers $[n]$. A node-incidence query is specified by a node index $i \in [n]$, and the query algorithm is returned the list of edges $(i,j)$ incident to $i$.

We start by considering an Erd\H{o}s-R\'{e}nyi random graph, as it is the simplest to analyse. We then proceed to random $d$-regular graphs. For the lower bounds, the task is to output a path between nodes $s=1$ and $t=n$. An algorithm for finding an $s$-$t$ path works as follows: In each step, the algorithm is allowed to ask one node-incidence query. We make no assumption about how the algorithm determines which query to make in each step, other than it being computable from all edges seen so far (the responses to the node-incidence queries). For randomized algorithms, the choice of query in each step is chosen randomly from a distribution over queries computable from all edges seen so far.

\subsection{Erd\H{o}s-R\'{e}nyi} 
Let $\bG$ be an Erd\H{o}s-R\'{e}nyi random graph, where each edge is present independently with probability $p \geq 1.5 \ln(n)/n$ and let $\Alg^\star$ be a possibly randomized algorithm for computing an $s$-$t$ path in $\bG$ when $s=1$ and $t=n$. Let $\alpha^\star$ be the probability that $\Alg^\star$ outputs a \emph{valid} $s$-$t$ path (all edges on the reported path are in $\bG$) and let $q$ be the worst case number of queries made by $\Alg^\star$ (for $\Alg^\star$ making an expected $q$ queries, we can always make it worst case $O(q)$ queries by decreasing $\alpha$ by a small additive constant). Here the probability is over both the random choices of $\Alg^\star$ and the random input graph $\bG$. By linearity of expectation, we may fix the random choices of $\Alg^\star$ to obtain a deterministic algorithm $\Alg$ that outputs a valid $s$-$t$ path with probability $\alpha \geq \alpha^\star$. It thus suffices to prove an upper bound on $\alpha$ for such deterministic $\Alg$.

For a graph $G$, let $\pi(G)$ denote the \emph{trace} of running the deterministic $\Alg$ on $G$. If $i_1(G),\dots,i_q(G)$ denotes the sequence of queries made by $\Alg$ on $G$ and $\cN_1(G),\dots,\cN_q(G)$ denotes the returned sets of edges, then 
\[
\pi(G) := (i_1(G), \cN_1(G), i_2(G), \dots, i_q(G), \cN_q(G)).
\]
Observe that if we condition on a particular trace $\tau=(i_1, N_1, i_2, \dots, i_q, N_q)$, then the distribution of $\bG$ conditioned on $\pi(\Alg,\bG)=\tau$ is the same as if we condition on the set of edges incident to $i_1,\dots,i_q$ being precisely $N_1,\dots,N_q$. This is because the algorithm $\Alg$ is deterministic and the execution of $\Alg$ is the same for all graphs $G$ with the same such sets of edges incident to  $i_1,\dots,i_q$. Furthermore, no graph $G$ with a different set of incident edges for $i_1,\dots,i_q$ will result in the trace $\tau$.

For a trace $\tau = (i_1,N_1,\dots,i_q,N_q)$, call the trace \emph{connected} if there is a path from $s$ to $t$ using the \emph{discovered} edges
\[
\bigcup_{j=1}^q N_j.
\]
Otherwise, call it \emph{disconnected}. Intuitively, if a trace is disconnected, then it is unlikely that $\Alg$ will succeed in outputting a valid path connecting $s$ and $t$ as it has to guess some of the edges along such a path. Furthermore, if $\Alg$ makes too few queries, then it is unlikely that the trace is connected. Letting $\Alg(G)$ denote the output of $\Alg$ on the graph $G$, we have for a random graph $\bG$ that
\[
\alpha = \Pr[\Alg(\bG)\textrm{ is valid}] \leq \Pr[\pi(\bG)\textrm{ is connected}] + \Pr[\Alg(\bG)\textrm{ is valid}\mid \pi(\bG)\textrm{ is disconnected}].
\]
We now bound the two quantities on the right hand side separately.

The simplest term to bound is 
\[
\Pr[\Alg(\bG)\textrm{ is valid}\mid \pi(\Alg,\bG)\textrm{ is disconnected}].
\]
For this, let $\tau = (i_1,N_1,\dots,i_q,N_q)$ be an arbitrary disconnected trace in the support of $\pi(\bG)$ when $\bG$ is an Erd\H{o}s-R\'{e}nyi random graph, where each edge is present with probability $p \geq 1.5 \ln(n)/n$. Observe that the output of $\Alg$ is determined from $\tau$. Since $\tau$ is disconnected, the path reported by $\Alg$ on $\tau$ must contain at least one edge $(u,v)$ where neither $u$ nor $v$ is among $\cup_{j} \{i_j\}$ or otherwise the output path is valid with probability $0$ conditioned on $\tau$. But conditioned on the trace $\tau$, every edge that is not connected to $\{i_1,\dots,i_q\}$ is present independently with probability $p$. We thus conclude 
\[
\Pr[\Alg(\bG)\textrm{ is valid}\mid \pi(\bG)=\tau] \leq p.
\]
Since this holds for every disconnected $\tau$, we conclude
\[
\Pr[\Alg(\bG)\textrm{ is valid}\mid \pi(\bG)\textrm{ is disconnected}] \leq p.
\]

Next we bound the probability that $\pi(\bG)$ is connected. For this, define for $1\leq k\leq q$
\[
\pi_k(G) := (i_1(G),\cN_1(G),i_2(G),\dots,i_k(G),\cN_k(G))
\]
as the trace of $\Alg$ on $G$ after the first $k$ queries. As for $\pi(G)$, we say that $\pi_k(G)$ is connected if there is a path from $s$ to $t$ using the discovered edges 
\[
E(\pi_k(G)) = \bigcup_{j=1}^k \cN_j(G)
\]
and that it is disconnected otherwise. We further say that $\pi_k(G)$ is \emph{useless} if it is both disconnected and $|E(\pi_k(G))| \leq 2pnk$. Since 
\[
\Pr[\pi_k(\bG) \textrm{ is disconnected}] \geq \Pr[\pi_k(\bG) \textrm{ is useless}]
\]
we focus on proving that $\Pr[\pi_k(\bG) \textrm{ is useless}]$ is large. For this, we lower bound
\[
\Pr[\pi_k(\bG) \textrm{ is useless} \mid \pi_{k-1}(\bG)\textrm{ is useless}].
\]
Note that the base case $\pi_0(\bG)$ is defined to be useless as $s$ and $t$ are not connected when no queries have been asked and also $|E(\pi_0(G))|=0 \leq 2pn 0 = 0$.
Let $\tau_{k-1} = (i_1,N_1,\dots,i_{k-1},N_{k-1})$ be any useless trace. The query $i_k = i_k(\bG)$ is uniquely determined when conditioning on $\pi_{k-1}(\bG) = \tau_{k-1}$ and so is the edge set $E_{k-1}=E(\pi_{k-1}(\bG))$. Furthermore, we know that $|E_{k-1}| \leq 2pn(k-1)$. We now bound the probability that the query $i_k$ discovers more than $2pn$ new edges. If $i_k$ has already been queried, no new edges are discovered and the probability is $0$. So assume $i_k \notin \{i_1,\dots,i_{k-1}\}$. Now observe that conditioned on $\pi_{k-1}(\bG)=\tau_{k-1}$, the edges $(i_k,i)$ where $i \notin \{i_1,\dots,i_{k-1}\}$ are independently included in $\bG$ with probability $p$ each. The number of new edges discovered is thus a sum of $m \leq n$ independent Bernoullis $\bX_1,\dots,\bX_m$ with success probability $p$. A Chernoff bound implies $\Pr[\sum_i \bX_i > (1+\delta)\mu] < (e^{\delta}/(1+\delta)^{1+\delta})^\mu$ for any $\mu \geq mp$ and any $\delta>0$. Letting $\mu = np$ and $\delta = 1$ gives
\[
\Pr[\sum_i \bX_i > 2np] < (e/4)^{np} < e^{-np/3}.
\]
Since we assume $p > 1.5 \ln (n)/n$, this is at most $1/\sqrt{n}$. 

We next bound the probability that the discovered edges $\cN_k(\bG)$ makes $s$ and $t$ connected in $E(\pi_k(\bG))$. For this, let $V_s$ denote the nodes in the connected component of $s$ in the subgraph induced by the edges $E_{k-1}$. Define $V_t$ similarly. We split the analysis into three cases. First, if $i_k \in V_s$, then $\cN_k(\bG)$ connects $s$ and $t$ if and only if one of the edges $\{i_k\} \times V_t$ is in $\bG$. Conditioned on $\pi_{k-1}(\bG) = \tau_{k-1}$, each such edge is in $\bG$ independently either with probability $0$, or with probability $p$ (depending on whether one of the end points is in $\{i_1,\dots,i_{k-1}\}$). A union bound implies that $s$ and $t$ are connected in $E(\pi_k(\bG))$ with probability at most $p|V_t|$. A symmetric argument upper bounds the probability by $p|V_s|$ in case $i_k \in V_t$. Finally, if $i_k$ is in neither of $V_s$ and $V_t$, it must have an edge to both a node in $V_s$ and in $V_t$ to connect $s$ and $t$. By independence, this happens with probability at most $p^2|V_t| |V_s|$. We thus conclude that 
\[
\Pr[\pi_k(\bG) \textrm{ is connected} \mid \pi_{k-1}(\bG) = \tau_{k-1}] \leq p \max\{|V_s|,|V_t|\} \leq p (|E_{k-1}|+1) \leq 2p^2 n k.
\]
A union bound implies
\[
\Pr[\pi_k(\bG)\textrm{ is useless} \mid \pi_{k-1}(\bG) \textrm{ is useless}] \geq 1-2p^2 nk - 1/\sqrt{n}.
\]
This finally implies
\begin{eqnarray*}
\Pr[\pi(\bG)\textrm{ is useless}] &=& \prod_{k=1}^q \Pr[\pi_k(\bG) \textrm{ is useless} \mid \pi_{k-1}(\bG) \textrm{ is useless}] \\
&\geq& 
\prod_{k=1}^q \left(1 - 2p^2 n k - 1/\sqrt{n}\right) \\
&\geq& 1-\sum_{k=1}^q (2p^2 n k + 1/\sqrt{n}) \\
&\geq& 1-p^2 n (q+1)^2 - q/\sqrt{n}.
\end{eqnarray*}
It follows that
\[
\Pr[\pi(\bG) \textrm{ is connected}]  = 1-\Pr[\pi(\bG) \textrm{ is disconnected}] \leq 1-\Pr[\pi(\bG) \textrm{ is useless}] \leq p^2 n (q+1)^2 + q/\sqrt{n}.
\]
For $q = o(1/(p\sqrt{n}))$ and $p \geq 1.5 \ln(n)/n$, this is $o(1)$. Note that for the lower bound to be meaningful, we need $p = O(1/\sqrt{n})$ as otherwise the bound on $q$ is less than $1$. (Indeed, 
for $p=\Omega(1/\sqrt n)$, $s$ and $t$ have a common neighbor with probability bounded away from $0$ and if so $2$ queries suffice). 
This concludes the proof of Theorem~\ref{thm:lbErdos}.

\subsection{$d$-Regular Graphs}
We now proceed to random $d$-regular graphs. Assume $dn$ is even, as otherwise a $d$-regular graph on $n$ nodes does not exist. Similarly to our proof for the Erd\H{o}s-R\'{e}nyi random graphs, we will condition on a trace of $\Alg$. Unfortunately, the resulting conditional distribution of a random $d$-regular graph is more cumbersome to analyse. We thus start by reducing to a slightly different problem. 

Let $\cM_{n,d}$ denote the set of all graphs on $nd$ nodes where the edges form a perfect matching on the nodes. There are thus $nd/2$ edges in any such graph. We think of the nodes of a graph $G \in \cM_{n,d}$ as partitioned into $n$ groups of $d$ nodes each, and we index the nodes by integer pairs $(i,j)$ with $i \in [n]$ and $j \in [d]$. Here $i$ denotes the index of the group. For a graph $G \in \cM_{n,d}$ and a sequence of group indices $p:= s,i_1,\dots,i_m,t$, we say that $p$ is a valid $s$-$t$ \emph{meta-path} in $G$, if for every two consecutive indices $a,b$ in $p$, there is at least one edge $((a,j_1), (b,j_2))$ in $G$. A meta-path is thus a valid path if and only if $s$ and $t$ are connected in the graph resulting from contracting the nodes in each group.

Now consider the problem of finding a valid $s$-$t$ meta-path in a graph $\bG$ drawn uniformly from $\cM_{n,d}$ (we write $\bG \sim \cM_{n,d}$ to denote such a graph) while asking \emph{group-incidence queries}. A group-incidence query is specified by a group index $i \in [n]$ and the answer to the query is the set of edges incident to the nodes $\{i\} \times \{1,\dots,d\}$.

We start by showing that an algorithm $\Alg^\star$ for finding an $s$-$t$ path in a random $d$-regular $n$-node graph, gives an algorithm $\Alg$ for finding an $s$-$t$ meta-path in a random $\bG \sim \cM_{n,d}$ using group-incidence queries.

\begin{lemma}
  \label{lem:matching}
  If there is a (possibly randomized) algorithm $\Alg^\star$ that reports a valid $s$-$t$ path with probability $\alpha$ in a random $d$-regular graph on $n$ nodes while making $q$ node-incidence queries, then there is a deterministic algorithm $\Alg$ that reports a valid $s$-$t$ meta-path with probability at least $\exp(-O(d^2)) \alpha$ in a random graph $\bG \sim \cM_{n,d}$ while making $q$ group-incidence queries.
\end{lemma}

\begin{proof}
  Given an algorithm $\Alg^\star$ that reports a valid $s$-$t$ path in a random $d$-regular graph on $n$ nodes with probability $\alpha$, we start by fixing its randomness to obtain a deterministic algorithm $\Alg'$ with the same number of queries that outputs a valid $s$-$t$ path with probability at least $\alpha$. Next, let $\bG \sim \cM_{n,d}$. Let $i_1 \in [n]$ be the first node that $\Alg'$ queries (which is independent of the input graph). Our claimed algorithm $\Alg$ for reporting an $s$-$t$ meta-path in $\bG$ starts by querying the group $i_1$. Upon being returned the set of edges $\{((i_1,1), (j_1,k_1)), \dots, ((i_1,d),(j_d,k_d))\}$ incident to $\{i_1\} \times \{1,\dots,d\}$, we \emph{contract} the groups such that each edge $((i_1,h), (j,k))$ is replaced by $(i_1,j)$. If this creates any duplicate edges or self-edges, $\Alg$ aborts and outputs an arbitrarily chosen $s$-$t$ meta-path. Otherwise, the resulting set of edges $\{(i_1,j_1),\dots,(i_1,j_d)\}$ is passed on to $\Alg'$ as the response to the first query $i_1$. The next query $i_2$ of $\Alg'$ is then determined and we again ask it as a group-incidence query on $\bG$ and proceed by contracting groups in the returned set of edges and passing the result to $\Alg'$ if there are no duplicate or self-edges. Finally, if we succeed in processing all $q$ queries of $\Alg'$ without encountering duplicate or self-edges, $\Alg$ outputs the $s$-$t$ path reported by $\Alg'$ as the $s$-$t$ meta-path.

  To see that this strategy has the claimed probability of reporting a valid $s$-$t$ meta-path, let $\bG^\star$ be the graph obtained from $\bG$ by contracting \emph{all} groups. Observe that if we condition on $\bG^\star$ being a simple graph (no duplicate edges or self-edges), then the conditional distribution of $\bG^\star$ is precisely that of a random $d$-regular graph on $n$ nodes. It is well-known~\cite{Bender1974TheAN,Bollobs1980APP, wormald_1980, wormald_1999} that the contracted graph $\bG^\star$ is indeed simple with probability at least $\exp(-O(d^2))$ and the claim follows.
\end{proof}

In light of Lemma~\ref{lem:matching}, we thus set out to prove lower bounds for deterministic algorithms that report an $s$-$t$ meta-path in a random $\bG \sim \cM_{n,d}$ using group-incidence queries.

Let $\Alg$ be a deterministic algorithm making $q$ group-incidence queries that reports a valid $s$-$t$ meta-path with probability $\alpha$ in a random $\bG \sim \cM_{n,d}$. Similarly to our proof for Erd\H{o}s-R\'{e}nyi graphs, we start by defining the trace of $\Alg$ on a graph $G \in \cM_{n,d}$. If $i_1(G),\dots,i_q(G) \in [n]$ denotes the sequence of group-incidence queries made by $\Alg$ on $G$ and $\cN_1(G),\dots,\cN_q(G)$ denotes the returned sets of edges, then for $1 \leq k \leq q$, we define
\[
  \pi_k(G) = (i_1(G),\cN_1(G),\dots,i_k(G),\cN_k(G)).
\]
We also let $\pi(G) := \pi_q(G)$ denote the full trace. Call a trace $\tau_k = (i_1,N_1,\dots,i_k,N_k)$ connected if there is a sequence of group indices $p:=s,i_1,\dots,i_m,t$ such that for every two consecutive indices $a,b$ in $p$, there is an edge $((a,h), (b,k))$ in $\cup_i N_i$. Otherwise, call the trace disconnected. Letting $\Alg(G)$ denote the output of $\Alg$ on the graph $G$, we have
\[
  \alpha = \Pr[\Alg(\bG)\textrm{ is valid}] \leq \Pr[\pi(\bG) \textrm{ is connected}] + \Pr[\Alg(\bG) \textrm{ is valid} \mid \pi(\bG) \textrm{ is disconnected}].
\]
We bound the two terms separately, starting with the latter. So let $\tau= (i_1,N_1,\dots,i_q,N_q)$ be a disconnected trace in the support of $\pi(\bG)$. The output meta-path $\Alg(\bG) = p = s,i_1,\dots,i_m,t$ of $\Alg$ is determined from $\tau$. Since $\tau$ is disconnected, there must be a pair of consecutive indices $a,b$ in $p$ such that there is no edge $((a,h), (b,k)) \in \cup_i N_i$. Fix such a pair $a,b$. We now consider two cases. First, if either $a$ or $b$ is among $i_1,\dots,i_q$, then all edges incident to that group are among $\cup_i N_i$ conditioned on $\pi(\bG)=\tau$. It thus follows that $p$ is a valid $s$-$t$ meta-path with probability $0$ conditioned on $\pi(\bG)=\tau$. Otherwise, neither of $a$ and $b$ are among $i_1,\dots,i_q$. The set of edges $\cup_i N_i$ specify at most $dq$ edges of the matching $\bG$. For any node whose matching edge is not specified by $\cup_i N_i$, the conditional distribution of its neighbor is uniform random among all other nodes whose matching edge is not in $\cup_i N_i$. For each of the $d^2$ possible edges $((a,h),(b,k))$ between the groups $a$ and $b$, there is thus a probability at most $1/(nd -1 - 2dq)$ that the edge is in $\bG$ conditioned on $\pi(\bG)=\tau$. A union bound over all $d^2$ such edges finally implies
\[
  \Pr[\Alg(\bG) \textrm{ is valid} \mid \pi(\bG)=\tau] \leq \frac{d^2}{nd-1-2dq}.
\]
Since this holds for every disconnected $\tau$, we conclude
\[
  \Pr[\Alg(\bG) \textrm{ is valid} \mid \pi(\bG) \textrm{ is disconnected}] \leq \frac{d^2}{nd-1-2dq}.
\]

Next, to bound $\Pr[\pi(\bG)\textrm{ is connected}]$, we show that
\[
\Pr[\pi_k(\bG) \textrm{ is disconnected} \mid \pi_{k-1}(\bG) \textrm{ is disconnected}]
\]
is large. So let $\tau_{k-1} = (i_1,N_1,\dots,i_{k-1}, N_{k-1})$ be a disconnected trace in the support of $\pi_{k-1}(\bG)$. The next query $i_k = i_k(\bG)$ of $\Alg$ is fixed conditioned on $\pi_{k-1}(\bG) = \tau_{k-1}$. We have a two cases. First, if $i_k \in \{i_1,\dots,i_{k-1}\}$ then no new edges are returned by the query and we conclude 
\[
\Pr[\pi_k(\bG) \textrm{ is disconnected} \mid \pi_{k-1}(\bG)=\tau_{k-1}] = 1.
\]
Otherwise, let $V_s$ denote the subset of group-indices $j$ for which there is a meta-path from $s$ to $j$. Similarly, let $V_t$ denote the subset of group-indices $j$ for which there is a meta-path from $t$ to $j$. We have $V_s \cap V_t = \emptyset$. Now if $i_k \in V_s$, we have that $\pi_k(\bG)$ is connected only if there is an edge between a node $(i_k, j)$ with $j \in [d]$ and a node $(b, k)$ with $b \in V_t$. Let $r \in \{0,\dots,d\}$ denote the number of nodes $(i_k,j)$ with $j \in [d]$ for which the corresponding matching edge is not in $\cup_i N_i$. Conditioned on $\pi_{k-1}(\bG) = \tau_{k-1}$, the neighbor of any such node is uniform random among all other nodes for which the corresponding matching edge is not in $\cup_i N_i$. There are at least $nd-1 - 2d(k-1)$ such nodes. A union bound over at most $rd|V_t| \leq d^2 |V_t|$ pairs $((i_k,j),(b, k))$ implies that $\pi_k(\bG)$ is connected with probability at most $d^2|V_t|/(nd - 1-2d(k-1))$. A symmetric arguments gives an upper bound of $d^2|V_s|/(nd-1-2d(k-1))$ in case $i_k \in V_t$. Finally, if $i_k$ is in neither of $V_s$ and $V_t$, then there must still be an edge $((i_k,j), (a,k))$ for a group $a \in V_s$. We thus conclude 
\[
\Pr[\pi_k(\bG) \textrm{ is connected} \mid \pi_{k-1}(\bG)=\tau_{k-1}] \leq \frac{d^2\max\{|V_s|, |V_t|\}}{nd-1-2d(k-1)} \leq \frac{d^3 k}{nd-1-2dq}.
\]
Since this holds for every disconnected trace $\tau_{k-1}$, we finally conclude
\[
\Pr[\pi(\bG) \textrm{ is disconnected}] \geq \prod_{k=1}^q \left(1-\frac{d^3 k}{nd-1-2dq}\right) \geq 1 - \sum_{k=1}^q \frac{d^3 k}{nd-1-2dq} \geq 1-\frac{d^3 q^2}{nd-1-2dq}
\]
and thus
\[
\Pr[\pi(\bG) \textrm{ is connected}] \leq \frac{d^3 q^2}{nd-1-2dq}.
\]
For constant degree $d$, if $q = o(\sqrt{n})$, this is $o(1)$. Together with Lemma~\ref{lem:matching}, we have thus proved Theorem~\ref{thm:lbRegular}.

\bibliography{refs}
\bibliographystyle{abbrv}
\end{document}